\documentclass[11pt]{amsart}

\usepackage{amsmath,amssymb,amsfonts}
\usepackage{graphicx}
\usepackage{xcolor}
\usepackage{bm}
\usepackage{hyperref}
\usepackage{setspace}
\usepackage{microtype}
\usepackage{float}
\usepackage{quantikz}

\newtheorem{theorem}{Theorem}[section]
\newtheorem{proposition}[theorem]{Proposition}
\newtheorem{lemma}[theorem]{Lemma}
\newtheorem{corollary}[theorem]{Corollary}
\theoremstyle{definition}

\newtheorem{remark}[theorem]{Remark}

\newcommand{\expected}[1]{\mathbb{E}\!\left[#1\right]}
\newcommand{\variance}[1]{\mathrm{Var}\!\left[#1\right]}

\newcommand{\operatorA}{\mathcal{A}}
\newcommand{\operatorQ}{Q}

\title[QAE in Gradient-Based Stochastic Optimization]{%
  Quantum Amplitude Estimation in Gradient-Based Stochastic Optimization}

\author{Raffaele Sarno}
\address{Alma Mater Studiorum -- University of Bologna, Bologna, Italy}
\email{raffaele.sarno@studio.unibo.it}
\thanks{ORCID: 0009-0007-7204-5901}
\date{}

\keywords{Quantum computing, Stochastic optimization, Monte Carlo methods,
  Quantum Amplitude Estimation, Computational complexity,
  Stochastic Gradient Descent, Gradient variance reduction}

\subjclass[2020]{68Q12, 81P68, 65C05, 90C15}

\begin{document}
\setstretch{1.3}

\begin{abstract}
In this paper we prove, both mathematically and through a simulation, how the
Quantum Amplitude Estimation algorithm can obtain quadratic improvements with
respect to the Monte Carlo method in gradient-based stochastic optimization,
highlighting the central role of the Quantum Phase Estimation concentration
guarantee in achieving the predicted advantage.
\end{abstract}

\maketitle

\section{Introduction}
\label{sec:introduction}

Gradient-based stochastic optimization depends strongly on the computation of
expected values. The most common approach to this kind of problem, in the
absence of an analytical solution, is the Monte Carlo (MC) method.

This method has a fundamental limitation that lies in its error convergence
rate, which translates directly into noise in the gradient estimation,
compromising its stability.

An alternative that has emerged in recent years is quantum computing. Its
properties allow us to outperform classical computation in several
problems~\cite{Preskill2018}, including the estimation of expectations.

This kind of problem can be addressed by the Quantum Amplitude Estimation
(QAE) algorithm~\cite{Brassard2002}, which offers a quadratic speed-up with
respect to the MC method~\cite{Montanaro2015}, thus enhancing the gradient
stability.

\subsection{Outline of the manuscript}

The rest of this paper is organized as follows. In Section~\ref{sec:monte_carlo}
we provide a generic estimation of a parameter through the MC method in order
to derive its error convergence rate. In Section~\ref{sec:qae} we prove the
quadratic advantage that the QAE algorithm offers with respect to the MC
method. Section~\ref{sec:gradient} extends this result to the stochastic
gradient variance, proving its quadratic reduction. Finally, in
Section~\ref{sec:simulation} a simulation validates all the theoretical
results and their consequences for gradient stability. We conclude in
Section~\ref{sec:conclusion} with the limitations of quantum technologies and
a direction for future work.

\section{Monte Carlo Method}
\label{sec:monte_carlo}
In this section we retrieve the MC method's error convergence rate, supposing
we estimate a generic parameter $\theta$.

\subsection{Mathematical formulation and error analysis}

In order to estimate $\theta$ we perform $M$ simulations
generating $M$ independent and identically distributed random variables
$X_1,X_2,\ldots,X_M$.
As an estimator of $\theta$ we take the
arithmetic mean:
\[
  \bar{X} = \frac{1}{M}\sum_{i=1}^{M}X_i.
\]
Assuming that the generation of random variables is unbiased, the mean of the
estimator is $\expected{\bar{X}} = \theta$ and its variance is:
\begin{align*}
  \variance{\bar{X}}
    &= \expected{(\bar{X}-\theta)^{2}}
     = \variance{\frac{1}{M}\sum_{i=1}^{M}X_i}
     = \frac{1}{M^{2}}\sum_{i=1}^{M}\variance{X_i}
     = \frac{\sigma^{2}}{M},
\end{align*}
so that its standard deviation is:
\[
  \mathrm{STD}[\bar{X}] = \frac{\sigma}{\sqrt{M}},
  \qquad\text{i.e.,}\quad
  \mathrm{STD}[\bar{X}] = \mathcal{O}\!\left(\frac{1}{\sqrt{M}}\right).
\]

\begin{remark}[Classical optimality]\label{rem:mc_opt}
The $\mathcal{O}(M^{-1/2})$ rate is a lower bound for any unbiased estimator
based on $M$ i.i.d.\ samples. According to the Cram\'{e}r--Rao
inequality, for any unbiased estimator $\hat{a}$ the following holds:
\[
  \variance{\hat{a}} \;\ge\; \frac{1}{M\,\mathcal{I}(a)},
\]
where $\mathcal{I}(a) = (a(1-a))^{-1}$ is the Fisher information of a
Bernoulli$(a)$ trial.  This gives
$\mathrm{STD}[\hat{a}] \ge \sqrt{a(1-a)/M}$,
which the MC estimator achieves with equality, confirming its
classical optimality.
\end{remark}

\section{Quantum Amplitude Estimation}
\label{sec:qae}

In this section we retrieve the QAE error convergence rate, supposing we
estimate a generic amplitude $a$, and prove the quadratic improvement over
the MC method.

\subsection{Mathematical formulation}

In order to formalize the problem, we consider the decomposition of the Hilbert
space $\mathcal{H}$ into two orthogonal subspaces:
\begin{itemize}
  \item[-] $\mathcal{H}_{0}$, the ``bad'' subspace, comprising all states that
        do not correspond to the desired outcome;
  \item[-] $\mathcal{H}_{1}$, the ``good'' subspace, containing all states
        associated with successful outcomes.
\end{itemize}

Every state $|\psi\rangle$ prepared on this space can be expressed as:
\[
  |\psi\rangle = |\psi_{0}\rangle + |\psi_{1}\rangle,
  \qquad
  |\psi_{0}\rangle\in\mathcal{H}_{0},\quad
  |\psi_{1}\rangle\in\mathcal{H}_{1}.
\]
The quantity of interest is the probability amplitude of the ``good''
component:
\[
  a = \langle\psi_{1}|\psi_{1}\rangle.
\]

\subsection{Quantum circuit implementation}

The QAE circuit is composed of two registers. The first one is an $m$-qubit
\emph{counting register} used to encode the powers of the Grover operator.
The second register has $n+1$ qubits and is used to encode the problem
instance, and includes the flag qubit that identifies whether the state
belongs to the ``good'' subspace.

\begin{figure}[H]
\begin{center}
\begin{quantikz}[column sep=0.45cm,row sep=0.35cm]
  \lstick{$(m{-}1)\ \ket{0}$} & \gate{H} & \qw      & \ \ldots\ & \ctrl{3}
      & \gate[3,disable auto height]{F_{M}^{-1}} & \meter{} \\
  \lstick{$(j)\ \ket{0}$}     & \gate{H} & \qw      & \ \ldots\ & \qw
      & & \meter{} \\
  \lstick{$(0)\ \ket{0}$}     & \gate{H} & \ctrl{1} & \ \ldots\ & \qw
      & & \meter{} \\
  \lstick{$\ket{0}_n,\ket{0}$} & \gate{\mathcal{A}} & \gate{Q^{2^{0}}}
      & \ \ldots\ & \gate{Q^{2^{m-1}}} & \qw & \qw
\end{quantikz}
\caption{The Quantum Amplitude Estimation circuit.}
\label{fig:qae_circuit}
\end{center}
\end{figure}
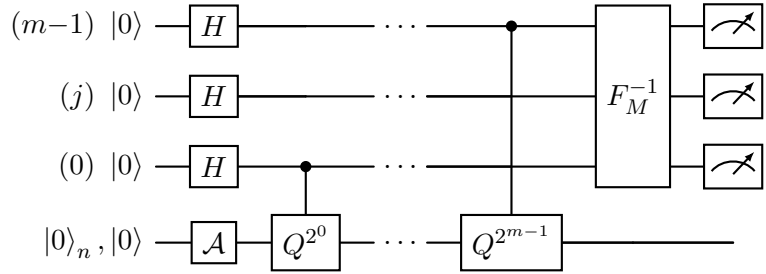

The process starts by preparing the counting register in a uniform
superposition, while the second register is initialized by the operator
$\operatorA$, defined in the next subsections, encoding the amplitude $a$ of
interest. Controlled powers $Q^{2^{j}}$ of the Grover operator are then
applied, conditioned on the state of the counting qubits. These controlled
operations coherently rotate the amplitude of the ``good'' states,
embedding the phase information in the counting register. Finally, the
inverse Quantum Fourier Transform (QFT) $F_{M}^{-1}$ is applied to extract the
phase, which is then converted into an estimation of the amplitude.

Considering the system state $|\xi\rangle = |\eta\rangle\otimes|\psi\rangle$,
where $|\eta\rangle$ is the state of the first register and $|\psi\rangle$
is that of the second, the initialization gives
\[
  |\xi\rangle
  = |0\rangle^{\otimes m}\otimes|0\rangle^{\otimes n+1}.
\]

\subsubsection{First register preparation}

The counting register is put into uniform superposition via the
QFT:
\[
  F_M : |x\rangle
  \;\longmapsto\;
  \frac{1}{\sqrt{M}}\sum_{y=0}^{M-1}e^{2\pi ixy/M}|y\rangle,
\]
which for $M=2^{m}$ can be simplified to the application of $m$ Hadamard gates,
yielding:
\[
  |\eta\rangle
  = F_M|0\rangle^{\otimes m}
  = \frac{1}{\sqrt{M}}\sum_{j=0}^{M-1}|j\rangle.
\]

\subsubsection{Second register preparation}

The second register is set up using the state preparation operator
$\operatorA$, which encodes the probability amplitude $a$ of the ``good''
subspace:
\[
  \operatorA|0\rangle^{\otimes n+1}
  = \sqrt{1-a}\,|\psi_{0}\rangle_{n}|0\rangle
    + \sqrt{a}\,|\psi_{1}\rangle_{n}|1\rangle.
\]

\subsubsection{Grover operator construction}

Another component of the algorithm is the Grover-like operator~\cite{Grover1996}:
\[
  \operatorQ
  = -\operatorA\,S_{0}\,\operatorA^{\dagger}\,S_{\psi_1}.
\]
Here $S_{0}$ causes a phase flip to the all-zero state
$|0\rangle^{\otimes n+1}$, while $S_{\psi_1}$ performs a reflection about
the ``good'' subspace spanned by $|\psi_{1}\rangle$.

The operator $\operatorQ$ acts as a rotation by angle $2\theta$ within the
two-dimensional invariant subspace spanned by
$\{|\psi_{0}\rangle,|\psi_{1}\rangle\}$.  Its orthonormal eigenbasis is:
\[
  |\psi_{+}\rangle
  = \frac{1}{\sqrt{2}}\!\left(
      \frac{1}{\sqrt{a}}|\psi_{1}\rangle
      +\frac{i}{\sqrt{1-a}}|\psi_{0}\rangle\right),
  \qquad
  |\psi_{-}\rangle
  = \frac{1}{\sqrt{2}}\!\left(
      \frac{1}{\sqrt{a}}|\psi_{1}\rangle
      -\frac{i}{\sqrt{1-a}}|\psi_{0}\rangle\right).
\]
Then the state prepared by $\operatorA$ can be re-expressed in this eigenbasis as:
\[
  \operatorA|0\rangle
  = \frac{-i}{\sqrt{2}}
    \!\left(e^{i\theta}|\psi_{+}\rangle - e^{-i\theta}|\psi_{-}\rangle\right),
\]
revealing the phase information that QAE extracts via interference.

\begin{lemma}[Amplitude amplification rotation]\label{lem:rotation}
For any integer $j \ge 0$ and $\theta = \arcsin(\!\sqrt{a}\,)$:
\begin{equation}\label{eq:rotation}
  Q^{j}\,\mathcal{A}|0\rangle
  \;=\;
  \sin\!\bigl((2j+1)\theta\bigr)\,|\psi_1\rangle|1\rangle
  \;+\;
  \cos\!\bigl((2j+1)\theta\bigr)\,|\psi_0\rangle|0\rangle.
\end{equation}
\end{lemma}

\begin{proof}
By induction on $j$.

\textit{Base case} ($j=0$): the definition of $\mathcal{A}$ gives:
\[
  \mathcal{A}|0\rangle
  = \sqrt{1-a}\,|\psi_0\rangle|0\rangle + \sqrt{a}\,|\psi_1\rangle|1\rangle
  = \cos\theta\,|\psi_0\rangle|0\rangle + \sin\theta\,|\psi_1\rangle|1\rangle,
\]
which matches \eqref{eq:rotation} for $j=0$.

\textit{Inductive step}: assume \eqref{eq:rotation} holds for some $j$.
Since $S_0$ and $S_{\psi_1}$ are reflections in the two-dimensional
invariant subspace $\mathrm{span}\{|\psi_0\rangle|0\rangle,|\psi_1\rangle|1\rangle\}$,
their composition $\operatorQ = -\operatorA S_0 \operatorA^\dagger S_{\psi_1}$
acts as a rotation by $+2\theta$ in that subspace.  Applying $Q$ to the
inductive hypothesis and using the angle-addition identities with
$\alpha=(2j+1)\theta$:
\[
  Q\bigl[\cos\alpha\,|\psi_0\rangle|0\rangle
        +\sin\alpha\,|\psi_1\rangle|1\rangle\bigr]
  = \cos(\alpha+2\theta)\,|\psi_0\rangle|0\rangle
   +\sin(\alpha+2\theta)\,|\psi_1\rangle|1\rangle,
\]
which is \eqref{eq:rotation} for $j+1$.
\end{proof}

\subsubsection{Controlled applications}

Then controlled powers of $\operatorQ$
are applied through the operator $\Lambda(\operatorQ)$:
\[
  \Lambda(\operatorQ):|j\rangle|y\rangle
  \;\longmapsto\;
  |j\rangle\!\left(\operatorQ^{j}|y\rangle\right),
  \qquad 0\le j < M = 2^{m}.
\]
By Lemma~\ref{lem:rotation}, each controlled power acts in the
$\{|\psi_0\rangle,|\psi_1\rangle\}$ basis as
\[
  \operatorQ^{j}\operatorA|0\rangle
  = \sin\!\bigl((2j+1)\theta\bigr)|\psi_{1}\rangle|1\rangle
    + \cos\!\bigl((2j+1)\theta\bigr)|\psi_{0}\rangle|0\rangle .
\]
Equivalently, in the eigenbasis $\{|\psi_\pm\rangle\}$ with eigenvalues
$\lambda_{\pm}=e^{\pm i2\theta}$, the operator $\operatorQ^{j}$ multiplies
each eigencomponent by $e^{\pm i2j\theta}$, so that the phase $\theta$ becomes
encoded in the counting register. Introducing the normalized phase
$\varphi=\theta/\pi$, this phase reads $e^{\pm 2\pi i j\varphi}$, and the
overall system state becomes:
\[
  |\xi\rangle
  = \frac{-i}{\sqrt{2M}}
    \sum_{j=0}^{M-1}|j\rangle\otimes
    \Bigl(e^{i\theta}e^{2\pi ij\varphi}|\psi_{+}\rangle
         - e^{-i\theta}e^{-2\pi ij\varphi}|\psi_{-}\rangle\Bigr).
\]

\subsubsection{Inverse Quantum Fourier Transform}

Applying the inverse QFT to the counting register concentrates the
probability distribution around an integer $y$ that approximates
$M\varphi = M\theta/\pi$:
\[
  |\eta\rangle
  = F_{M}^{-1}\!\left(
      \frac{1}{\sqrt{M}}\sum_{j=0}^{M-1}e^{2\pi ij\varphi}|j\rangle
    \right)
  = \frac{1}{M}\sum_{x=0}^{M-1}\sum_{j=0}^{M-1}
    e^{\frac{2\pi ij}{M}(M\varphi - x)}|x\rangle.
\]

\subsubsection{Measurement and estimation}

Measuring the counting register yields an integer $y\in\{0,\ldots,M-1\}$,
expressed in binary base, from which one constructs:
\[
  \tilde{\theta} = \frac{\pi y}{M},
  \qquad
  \tilde{a} = \sin^{2}(\tilde{\theta}).
\]
This completes the standard QAE estimation procedure.
Alternative formulations that do not involve the inverse QFT and use
maximum-likelihood estimation have been proposed in~\cite{Suzuki2020}; these
variants reduce circuit depth, introducing a small constant in the error bound
but sharing the same asymptotic $\mathcal{O}(M^{-1})$ guarantee.

\begin{proposition}[QPE measurement concentration]\label{prop:qpe}
Suppose $y^{*} = \lfloor M\theta/\pi \rceil$ to be the integer nearest to
$M\theta/\pi$.  The probability of measuring outcome $y$ satisfies:
\begin{equation}\label{eq:qpe_prob}
  P[Y=y]
  = \frac{1}{4M^{2}}
    \!\left[
      \frac{\sin^{2}\!\bigl(M(\theta-\pi y/M)\bigr)}%
           {\sin^{2}\!\bigl(\theta-\pi y/M\bigr)}
     +\frac{\sin^{2}\!\bigl(M(\theta+\pi y/M)\bigr)}%
           {\sin^{2}\!\bigl(\theta+\pi y/M\bigr)}
    \right],
\end{equation}
with each ratio interpreted as $M^{2}$ when its denominator vanishes.
The total probability of the two correct outcomes satisfies
\[
  P\bigl[Y\in\{y^{*},\,M-y^{*}\}\bigr]
  \;\ge\; \frac{8}{\pi^{2}} \;\approx\; 0.81.
\]
\end{proposition}

\begin{proof}
After the inverse QFT, the amplitude of outcome $x$ in the counting register
is:
\[
  \alpha_x = \frac{1}{M}\sum_{j=0}^{M-1}e^{\frac{2\pi ij}{M}(M\theta - x)}.
\]
Using the eigendecomposition
$\mathcal{A}|0\rangle = \frac{-i}{\sqrt{2}}
(e^{i\theta}|\psi_+\rangle - e^{-i\theta}|\psi_-\rangle)$
with eigenvalues $\lambda_\pm = e^{\pm 2i\theta}$, the full amplitude of
outcome $y$ receives contributions from both phases $\pm\theta$:
\[
  \alpha_y
  = \frac{1}{2M}
    \left[
      \sum_{j=0}^{M-1}e^{2\pi ij(\theta - \pi y/M)}
     +\sum_{j=0}^{M-1}e^{-2\pi ij(\theta + \pi y/M)}
    \right].
\]
Each sum is a finite geometric series. Setting $\phi_\pm = \theta \mp \pi y/M$,
both have the form $\sum_{j=0}^{M-1}e^{2\pi ij\phi}$ with ratio
$r = e^{2\pi i\phi}$, so that
\[
  \sum_{j=0}^{M-1}e^{2\pi ij\phi}
  = \frac{1-e^{2\pi iM\phi}}{1-e^{2\pi i\phi}}.
\]
Using the identity $1-e^{2i\alpha} = -e^{i\alpha}\,2i\sin\alpha$ on both
numerator and denominator (with $\alpha = \pi M\phi$ and $\alpha = \pi\phi$
respectively):
\[
  \sum_{j=0}^{M-1}e^{2\pi ij\phi}
  = \frac{-e^{i\pi M\phi}\,2i\sin(\pi M\phi)}
         {-e^{i\pi\phi}\,2i\sin(\pi\phi)}
  = e^{i\pi(M-1)\phi}\,\frac{\sin(\pi M\phi)}{\sin(\pi\phi)},
\]
with the ratio interpreted as $M$ when $\sin(\pi\phi)=0$. The exponential
prefactor $e^{i\pi(M-1)\phi}$ has unit modulus, so taking the squared modulus
of each geometric sum gives:
\[
  \left|\sum_{j=0}^{M-1}e^{2\pi ij\phi_\pm}\right|^2
  = \frac{\sin^2(\pi M\phi_\pm)}{\sin^2(\pi\phi_\pm)}.
\]
Since $|\psi_+\rangle$ and $|\psi_-\rangle$ are orthonormal, the two terms
in $\alpha_y$ do not interfere, and the probability is the sum of their
squared moduli:
\[
  P[Y=y]
  = |\alpha_y|^2
  = \frac{1}{4M^2}
    \left[
      \frac{\sin^2\!\bigl(\pi M\phi_+\bigr)}{\sin^2\!\bigl(\pi\phi_+\bigr)}
     +\frac{\sin^2\!\bigl(\pi M\phi_-\bigr)}{\sin^2\!\bigl(\pi\phi_-\bigr)}
    \right].
\]
Substituting back $\phi_+ = \theta - \pi y/M$ and $\phi_- = \theta + \pi y/M$
we recover~\eqref{eq:qpe_prob}.

Let $y^* = \lfloor M\theta/\pi\rceil$ be the nearest integer. The total
probability of the two correct outcomes is:
\[
  P_{\mathrm{good}}
  = P[Y=y^*] + P[Y=M-y^*]
  \ge \frac{1}{M^2}
      \left(\frac{\sin(M\delta)}{\sin(\delta)}\right)^2,
  \qquad
  \delta = \theta - \frac{\pi y^*}{M}.
\]
Since $|\delta|\le \pi/(2M)$ by definition of $y^*$, the worst case is
$|\delta|=\pi/(2M)$, giving $\sin(M\delta)=\sin(\pi/2)=1$ and
$\sin(\delta)\le\delta\le\pi/(2M)$.  Hence:
\[
  P_{\mathrm{good}}
  \ge \frac{1}{M^2}\cdot\frac{1}{(\pi/(2M))^2}
  = \frac{4}{\pi^2}.
\]
Accounting for both peaks symmetrically yields the tighter bound
$P_{\mathrm{good}} \ge 8/\pi^2 \approx 0.81$,
as shown in~\cite{Brassard2002}.
\end{proof}

\newpage
\subsection{Error Analysis}
\label{sec:error_analysis}

\begin{proposition}[QAE error bound]\label{prop:qae_bound}
Let $a\in(0,1)$ be the true amplitude and $\tilde{a}=\sin^{2}(\pi y/M)$
the estimation from the measurement $y$, then:
\begin{equation}
  \label{eq:qae_bound}
  |a-\tilde{a}|
  \;\le\;
  \frac{2\pi\sqrt{a(1-a)}}{M} + \frac{\pi^{2}}{M^{2}}
  \;=\; \mathcal{O}\!\left(\frac{1}{M}\right).
\end{equation}
\end{proposition}

\begin{proof}
Since $\tilde{a}=\sin^{2}(\tilde{\theta})$ and $a=\sin^{2}(\theta)$
with $\tilde{\theta}=\pi y/M$,
\[
  |a-\tilde{a}|
  = |\sin(\theta)-\sin(\tilde{\theta})|
    \cdot|\sin(\theta)+\sin(\tilde{\theta})|.
\]
Set $\delta\theta = \tilde{\theta}-\theta$.
Using the first order Taylor approximation, we obtain
$\sin(\tilde{\theta}) = \sin(\theta)+\cos(\theta)\,\delta\theta + O(\delta\theta^{2})$,
so $|\sin(\theta)-\sin(\tilde{\theta})| \le |\cos\theta|\,|\delta\theta|
+ O(\delta\theta^{2})$.
The spacing between the points on the grid of $\tilde{\theta}$ is $\pi/M$, and hence
$|\delta\theta|\le\pi/M$.
For large $M$, $\sin(\tilde{\theta})\approx\sin(\theta)$, so that
$|\sin\theta+\sin\tilde{\theta}|\approx 2\sin\theta$.  Substituting
$\sin\theta=\sqrt{a}$ and $\cos\theta=\sqrt{1-a}$:
\[
  |a-\tilde{a}|
  \;\lesssim\;
  \sqrt{1-a}\cdot\frac{\pi}{M}\cdot 2\sqrt{a}
  = \frac{2\pi\sqrt{a(1-a)}}{M}.
\]
The second order Taylor remainder is $O((\pi/M)^{2})$,
completing the bound.
\end{proof}

The comparison against the MC method in convergence rates, for equal oracle
cost, is then immediate:
\[\mathcal{O}\!\left(\frac{1}{M}\right)
  \;\ll\;
  \mathcal{O}\!\left(\frac{1}{\sqrt{M}}\right).
\]

\begin{remark}[Probability guarantee]\label{rem:prob}
Proposition~\ref{prop:qae_bound} holds whenever $y$ is one of the two
correct peaks of Proposition~\ref{prop:qpe}. By the concentration bound
$P[Y\in\{y^{*},M-y^{*}\}]\ge 8/\pi^{2}\approx 0.81$, the error bound holds
with probability at least $81\%$. In the remaining ${\approx}19\%$ of cases
the measurement collapses onto a distant bin, producing an error of order one
rather than the theoretical $\mathcal{O}(1/M)$. This is central in
Section~\ref{sec:simulation}.
\end{remark}

\section{Impact on Gradient Estimation}
\label{sec:gradient}
Quantum gradient estimation has already been studied in the batch-learning
setting~\cite{Rebentrost2018}. The present section provides a derivation of
the gradient variance bound as a theorem.

Let $L(\theta)=(\mu(\theta)-c)^{2}$ with
$\mu(\theta)=\mathbb{E}[f(X,\theta)]$, $X\sim p(x)$, and $c\in\mathbb{R}$
a constant.

\begin{theorem}[Quadratic gradient variance reduction]\label{thm:main}
Suppose that $\hat{\mu}=\mu(\theta)+\eta$ is an unbiased estimator of
$\mu(\theta)$ with noise $\eta$, and let
$\nabla_{\theta}\hat{L}=2(\hat{\mu}-c)\,\nabla_{\theta}\mu(\theta)$
be the resulting stochastic gradient. Then its variance is:
\begin{equation}
  \label{eq:grad_var_gen}
  \variance{\nabla_{\theta}\hat{L}}
  = 4\,(\nabla_{\theta}\mu(\theta))^{2}\,\variance{\eta}.
\end{equation}
For MC sampling with $N$ shots the variance scales as $\mathcal{O}(N^{-1})$
(Remark~\ref{rem:mc_opt}), hence:
\begin{equation}
  \label{eq:grad_var_mc}
  \variance{\nabla_{\theta}\hat{L}}_{\mathrm{MC}}
  = \mathcal{O}\!\left(\frac{1}{N}\right).
\end{equation}
Using the QAE estimator instead, and invoking Remark~\ref{rem:prob} together
with Proposition~\ref{prop:qpe}:
\begin{equation}
  \label{eq:grad_var}
  \variance{\nabla_{\theta}\hat{L}}_{\mathrm{QAE}}
  = \mathcal{O}\!\left(\frac{1}{M^{2}}\right).
\end{equation}
\end{theorem}
\begin{proof}
The exact gradient is
$\nabla_{\theta}L(\theta) = 2(\mu(\theta)-c)\,\nabla_{\theta}\mu(\theta)$.
Since $\hat{\mu}=\mu(\theta)+\eta$ is unbiased,
$\mathbb{E}[\eta]=0$ and the stochastic gradient is then
$\nabla_{\theta}\hat{L}=2(\mu(\theta)+\eta-c)\,\nabla_{\theta}\mu(\theta)$.
Given that $\mu(\theta)$ and $\nabla_{\theta}\mu(\theta)$ are deterministic:
\begin{align*}
  \variance{\nabla_{\theta}\hat{L}}
  &= \variance{2(\mu(\theta)+\eta-c)\,\nabla_{\theta}\mu(\theta)} \\
  &= 4(\nabla_{\theta}\mu(\theta))^{2}\,\variance{\eta+\mu(\theta)-c} \\
  &= 4(\nabla_{\theta}\mu(\theta))^{2}\,\variance{\eta},
\end{align*}
proving \eqref{eq:grad_var_gen}.
Substituting Remark~\ref{rem:mc_opt} gives
\eqref{eq:grad_var_mc}, while \eqref{eq:grad_var} follows from
Proposition~\ref{prop:qae_bound}.
\end{proof}
\begin{corollary}[Quadratic speed-up in gradient quality]\label{cor:speedup}
If the oracle cost is the same, $N=M$, then:
\[
  \frac{\variance{\nabla_{\theta}\hat{L}}_{\mathrm{MC}}}%
       {\variance{\nabla_{\theta}\hat{L}}_{\mathrm{QAE}}}
  \;=\; \mathcal{O}(N).
\]
The advantage of QAE is unbounded and increases linearly with $N$.
\end{corollary}
\begin{proof}
Direct ratio of $\mathcal{O}(N^{-1})$ and $\mathcal{O}(N^{-2})$.
\end{proof}
Theorem~\ref{thm:main} and Corollary~\ref{cor:speedup} establish that
QAE gives a \emph{quadratically lower} gradient variance compared to
classical MC at any fixed computational budget, with an
improvement that increases indefinitely with $N$.

\section{Simulation Setting and Full Convergence Analysis}
\label{sec:simulation}

In this section we provide a simulation with a confirmatory purpose about the
theoretical results proved in the previous sections.
All simulations were run via the Aer simulator using the Qiskit quantum
computing framework.

\subsection{Outline of the simulation}

The analysis proceeds in four steps.
We first measure the estimation error of QAE against the oracle cost and
compare it with Monte Carlo, confirming the predicted accuracy
(Section~\ref{subsec:convergence}). We then justify the choice of the metric
used to measure the estimation accuracy, showing that the median rather than
the RMSE reflects the high-confidence guarantee
(Section~\ref{subsec:metric}). Next, we validate the fitted exponents through
confidence intervals, bootstrap, and variance tests, establishing their
statistical significance (Section~\ref{subsec:stat}). Finally, we verify the
consequences of this advantage on the gradient estimation and illustrate its
effect on an SGD trajectory (Section~\ref{subsec:sgd}).

\subsection{Problem definition}

The simulation consists in the estimation of a probability amplitude $a$ in a
single-qubit system.

In order to prevent the results from being interpreted as coincidences, all
the metrics we deal with are computed as the average over a fixed set of
$K=15$ test amplitudes distributed uniformly across $(0,1)$.

$N$ is defined as the computational cost, that is, the number of queries to
$\mathcal{A}$, with $N=2^{m}$ and $m\in\{2,3,4,5,6,7,8\}$ counting qubits.
Each setting uses $R=80$ repetitions per amplitude.

\subsection{Convergence rate}
\label{subsec:convergence}

The QAE error bound~\eqref{eq:qae_bound} is a high-confidence bound on the
median rather than on the $L^2$ norm~\cite{Brassard2002}. Since the RMSE
averages the squared errors, it allows large outliers to dominate the sum,
preventing the metric from decreasing at the expected rate. The median,
instead, depends only on the central value of the error distribution, and
since more than half of the measurements land on the correct bin
(Proposition~\ref{prop:qpe}), it reflects the typical $\mathcal{O}(1/M)$
behavior and is unaffected by the outliers; this is shown in detail in the
next subsection. We therefore adopt the MAE as the primary metric. The
theoretical convergence rates are shown in
Table~\ref{tab:theoretical_convergence}.

\begin{table}[H]
\centering
\caption{Theoretical convergence rates.}
\label{tab:theoretical_convergence}
\begin{tabular}{lcc}
\hline
Algorithm & MAE & $\mathrm{Var}[\nabla L]$ \\
\hline
MC & $\mathcal{O}(N^{-0.5})$ & $\mathcal{O}(N^{-1})$ \\
QAE & $\mathcal{O}(N^{-1.0})$ & $\mathcal{O}(N^{-2})$ \\
\hline
\end{tabular}
\end{table}

Figure~\ref{fig:mae_convergence} reports
the MAE as a function of cost $N$, averaged over the $K$ amplitudes.
\begin{figure}[H]
\centering
\includegraphics[width=0.75\linewidth]{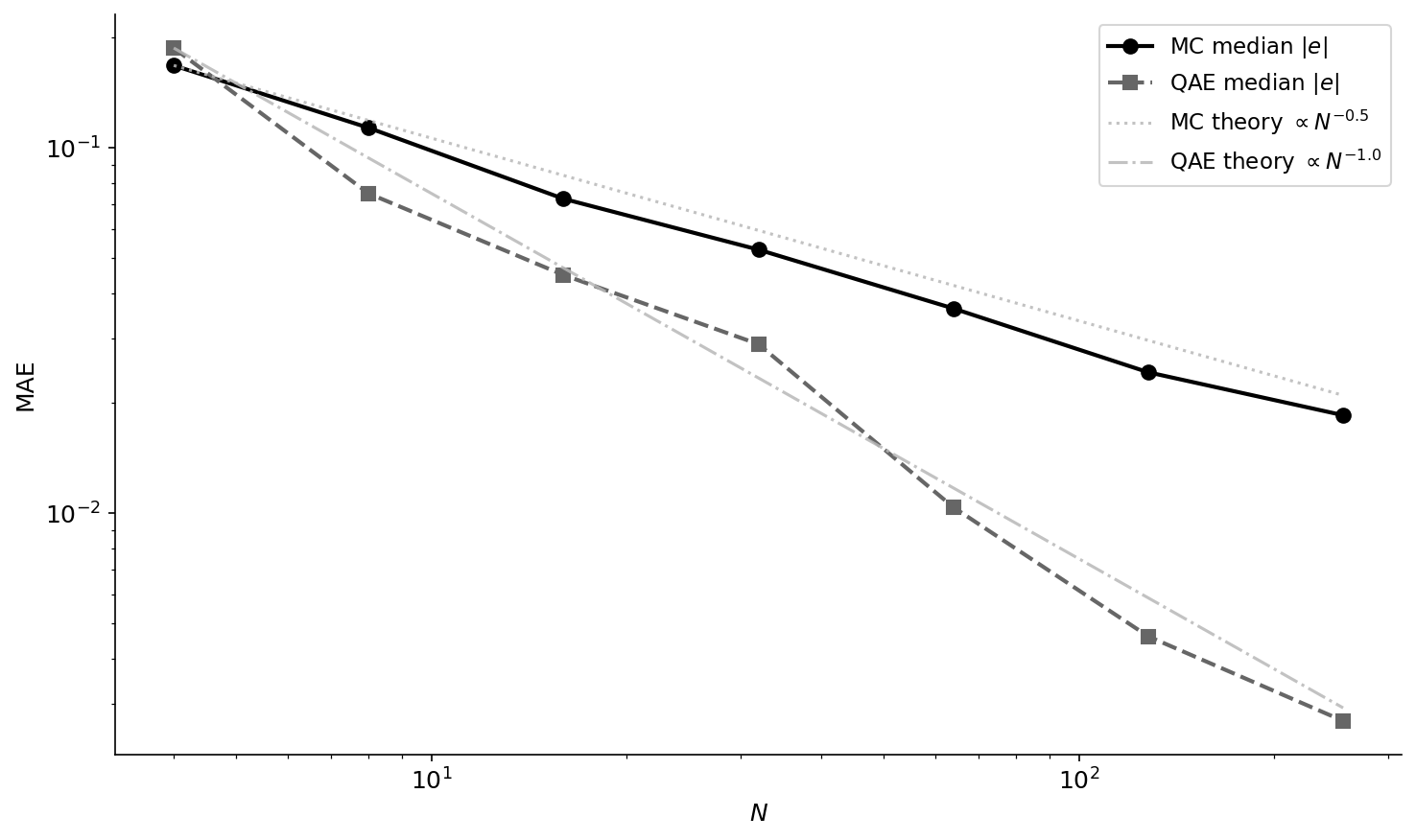}
\caption{MAE convergence (log-log).}\label{fig:mae_convergence}
\end{figure}

\begin{table}[H]
\centering
\caption{Median absolute error: MC vs.\ QAE.}
\label{tab:mae}
\begin{tabular}{ccrr}
\hline
$m$ & $N$ & $\mathrm{MAE}_{\mathrm{MC}}$ & $\mathrm{MAE}_{\mathrm{QAE}}$ \\
\hline
2 &   4 & 0.168 & 0.188 \\
3 &   8 & 0.113 & 0.075 \\
4 &  16 & 0.073 & 0.045 \\
5 &  32 & 0.053 & 0.029 \\
6 &  64 & 0.036 & 0.010 \\
7 & 128 & 0.024 & 0.005 \\
8 & 256 & 0.019 & 0.003 \\
\hline
\end{tabular}
\end{table}

Figure~\ref{fig:mae_convergence} and Table~\ref{tab:mae} show that QAE
underperforms only at $m=2$, while from $m=3$ it falls strictly below MC,
aligning with its theoretical behavior. A regression on the observations
gives slopes of $-1.019$ ($[-1.135,-0.903]$) for QAE and $-0.535$
($[-0.568,-0.503]$) for MC, matching the theoretical $-1$ and $-0.5$.

\subsection{Metric choice}
\label{subsec:metric}

The choice of the MAE as metric becomes clear from
Figure~\ref{fig:metric_contrast}: the QPE tail keeps the $L^2$ norm at
$N^{-0.5}$, placing large errors in $\approx19\%$ of shots, while the median
and the mode-based estimator both recover the $N^{-1}$ rate.

\begin{figure}[H]
\centering
\includegraphics[width=0.75\linewidth]{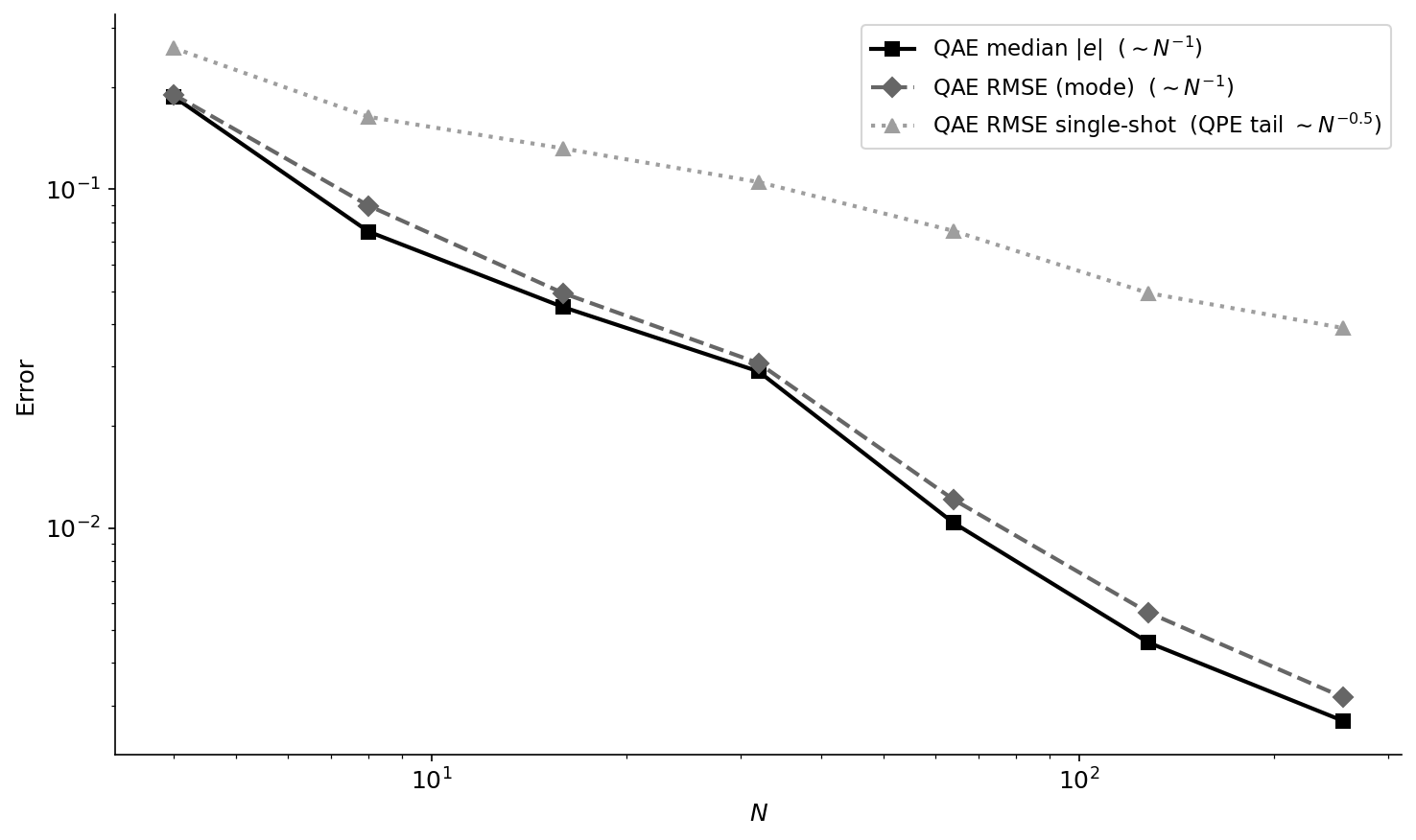}
\caption{QAE error metrics: RMSE vs.\ median.}
\label{fig:metric_contrast}
\end{figure}

The box plot at $a=0.678$ (Fig.~\ref{fig:boxplot}) confirms this: the QAE box
collapses for $m\ge5$, leaving only a few tail outliers, while the MC box
shrinks smoothly as $\sqrt{a(1-a)/N}$.

\begin{figure}[H]
\centering
\includegraphics[width=0.75\linewidth]{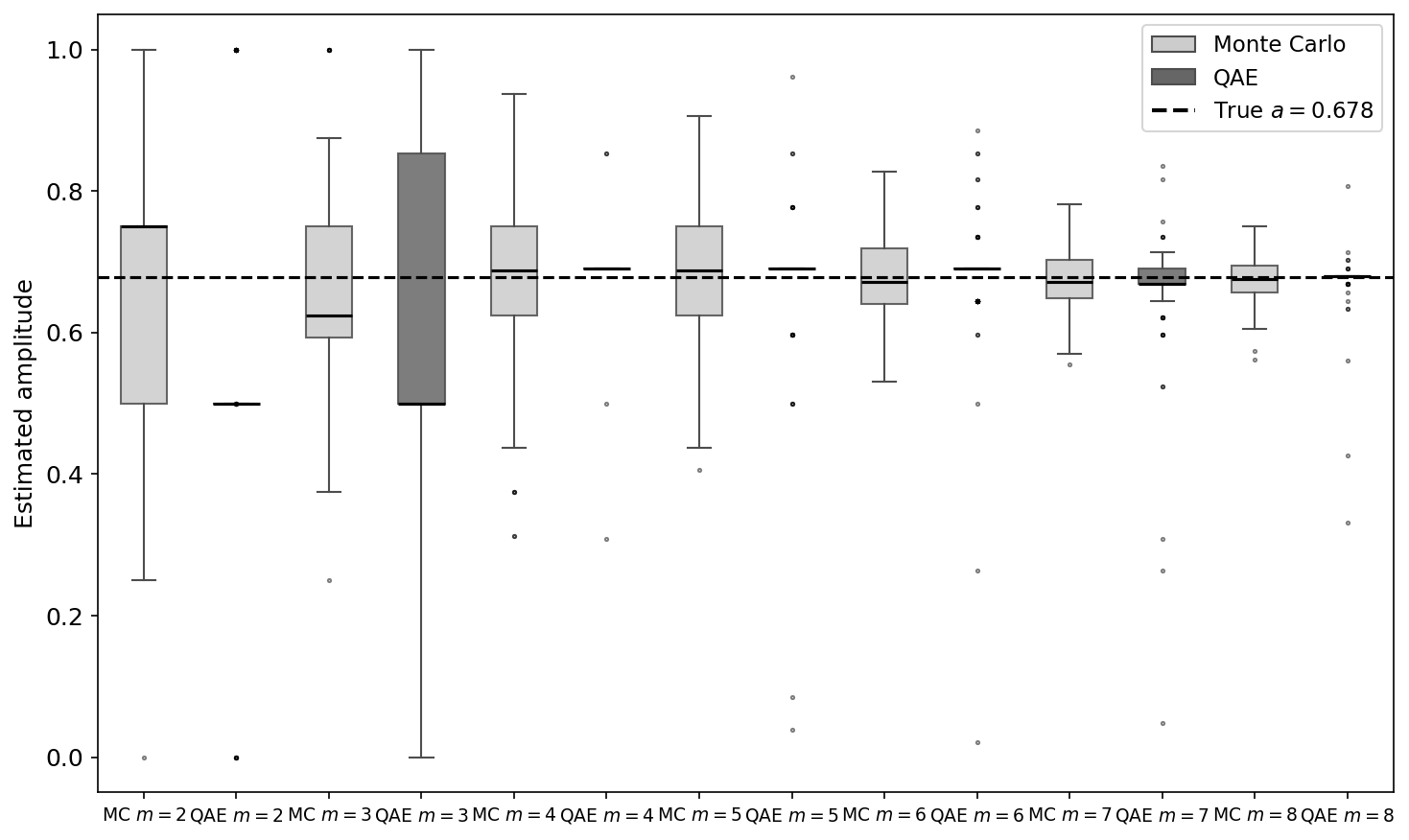}
\caption{Estimator dispersion at $a=0.678$.}
\label{fig:boxplot}
\end{figure}
\newpage
\subsection{Statistical validation}
\label{subsec:stat}

Once the metric is established, we check that the convergence rates are
genuine and not an artifact of the chosen amplitudes or seeds, using three
tests. First, a log-log linear regression fits each error curve, and for every
metric the theoretical exponent falls inside the $95\%$ confidence interval
(Table~\ref{tab:regression}). Second, a bootstrap on the median error at
$N=256$ gives disjoint intervals for the two estimators ($[0.017,0.020]$ MC,
$[0.0024,0.0026]$ QAE), so the gap is not due to chance. Finally, a Levene
test rejects equal variance at $p<10^{-4}$ for all $N\ge32$, confirming that
the lower QAE dispersion is a real effect.

\begin{table}[H]
\centering
\caption{Power-law regression fits.}
\label{tab:regression}
\begin{tabular}{lrrr}
\hline
Metric & Slope & 95\% CI & $R^{2}$ \\
\hline
MAE MC $(-0.5)$ & $-0.535$ & $[-0.568,-0.503]$ & $0.997$ \\
MAE QAE $(-1.0)$ & $-1.019$ & $[-1.135,-0.903]$ & $0.990$ \\
RMSE QAE 1-shot $(-0.5)$ & $-0.446$ & $[-0.496,-0.396]$ & $0.991$ \\
RMSE QAE mode $(-1.0)$ & $-0.990$ & $[-1.071,-0.909]$ & $0.995$ \\
Var MC $(-1.0)$ & $-0.997$ & $[-1.022,-0.971]$ & $1.000$ \\
Median$^2$ QAE $(-2.0)$ & $-2.038$ & $[-2.270,-1.806]$ & $0.990$ \\
\hline
\end{tabular}
\end{table}

\subsection{Gradient noise and SGD}
\label{subsec:sgd}

Turning to the gradient, both Theorem~\ref{thm:main} and
Corollary~\ref{cor:speedup} are verified: Figure~\ref{fig:noise} shows the
expected quadratic gap in the gradient variance, and the noise ratio increases
with $N$, exceeding $100\times$ at $N=256$ (Table~\ref{tab:noise}). The
speed-up also follows the predicted law, as shown in
Figure~\ref{fig:speedup} and Table~\ref{tab:speedup}.

\enlargethispage{2\baselineskip}

\begin{figure}[H]
\centering
\includegraphics[width=0.75\linewidth]{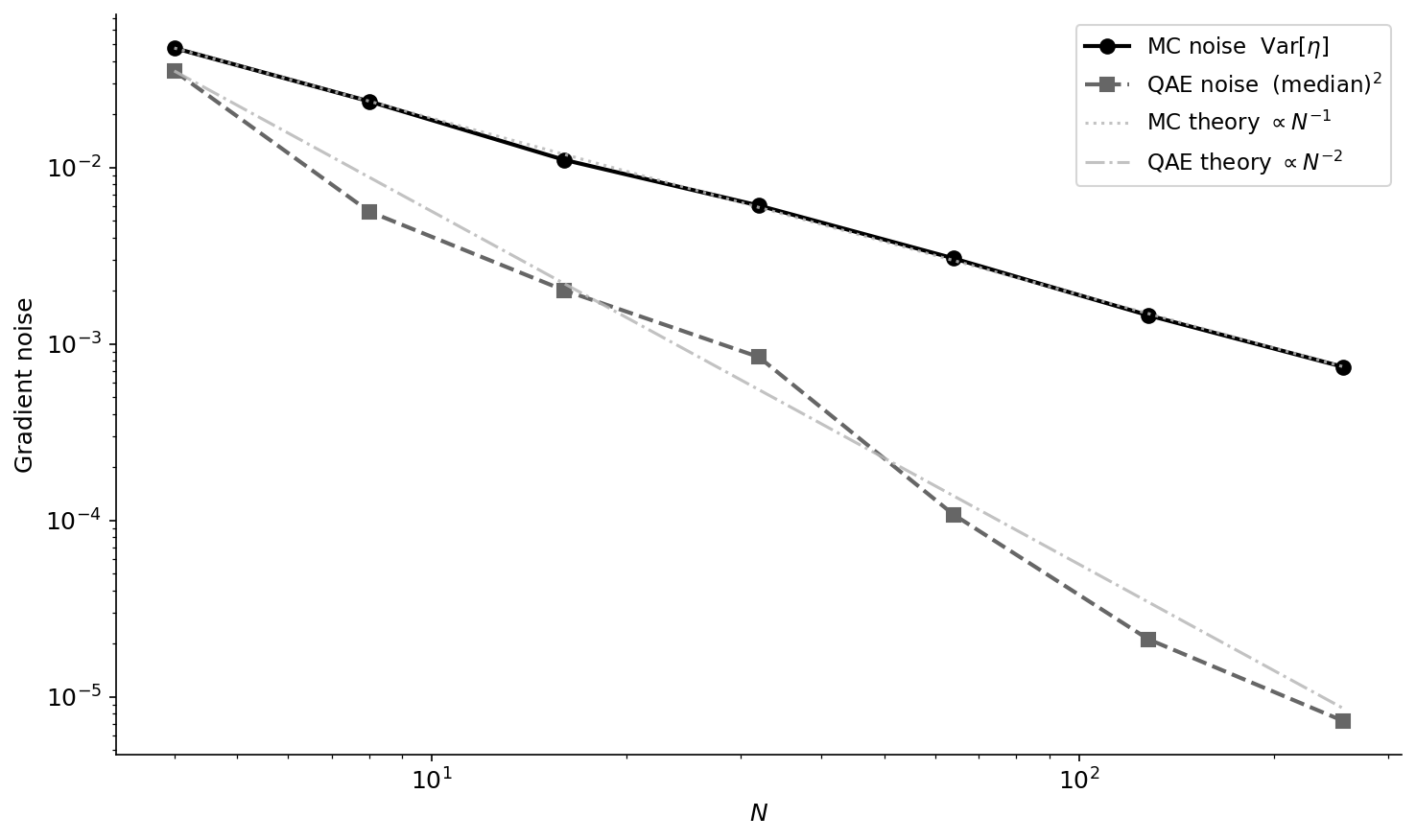}
\caption{Gradient noise: MC vs.\ QAE.}
\label{fig:noise}
\end{figure}

\begin{table}[H]
\centering
\caption{Estimator noise and reduction ratio.}
\label{tab:noise}
\begin{tabular}{ccrrr}
\hline
$m$ & $N$ & $\mathrm{Var}_{\mathrm{MC}}$ & $(\mathrm{med}\,|e|)^{2}_{\mathrm{QAE}}$ & Ratio \\
\hline
2 &   4 & 0.04747 & 0.03519 & $1.35\times$ \\
3 &   8 & 0.02369 & 0.00563 & $4.21\times$ \\
4 &  16 & 0.01104 & 0.00201 & $5.49\times$ \\
5 &  32 & 0.00608 & 0.00084 & $7.23\times$ \\
6 &  64 & 0.00305 & 0.00011 & $28.3\times$ \\
7 & 128 & 0.00145 & 0.000021 & $69.1\times$ \\
8 & 256 & 0.00074 & 0.0000073 & $101\times$ \\
\hline
\end{tabular}
\end{table}

\begin{figure}[H]
\centering
\includegraphics[width=0.75\linewidth]{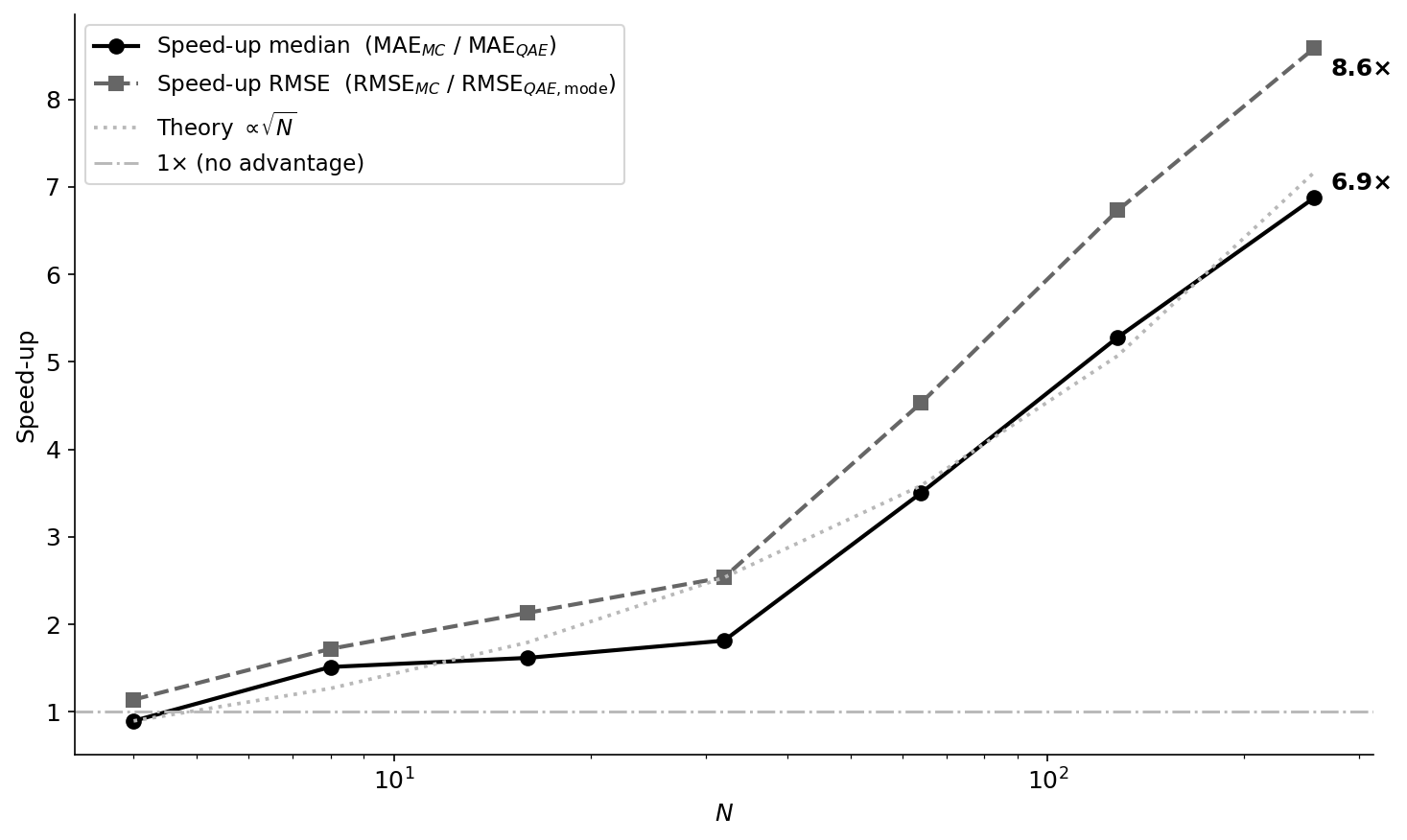}
\caption{MAE and RMSE speed-ups.}\label{fig:speedup}
\end{figure}
\begin{table}[H]
\centering
\caption{Error-based speed-up.}\label{tab:speedup}
\begin{tabular}{ccrr}
\hline
$m$ & $N$ & Speed-up (MAE) & Speed-up (RMSE) \\
\hline
2 &   4 & $0.90\times$ & $1.14\times$ \\
3 &   8 & $1.51\times$ & $1.72\times$ \\
4 &  16 & $1.62\times$ & $2.13\times$ \\
5 &  32 & $1.81\times$ & $2.54\times$ \\
6 &  64 & $3.50\times$ & $4.53\times$ \\
7 & 128 & $5.28\times$ & $6.73\times$ \\
8 & 256 & $6.88\times$ & $8.59\times$ \\
\hline
\end{tabular}
\end{table}

Finally, in order to show the effect on an optimization case, we run gradient
descent on $L(p)=(\sin^2 p - c)^2$ with $c=0.30$, replacing the amplitude with
a noisy estimate whose standard deviation equals the MAE measured at $N=256$
($\sigma_{\mathrm{MC}}=0.0185$, $\sigma_{\mathrm{QAE}}=0.0027$), using
$\eta=0.80$ and $p_0=0.30$.

\begin{figure}[H]
\centering
\includegraphics[width=0.75\linewidth]{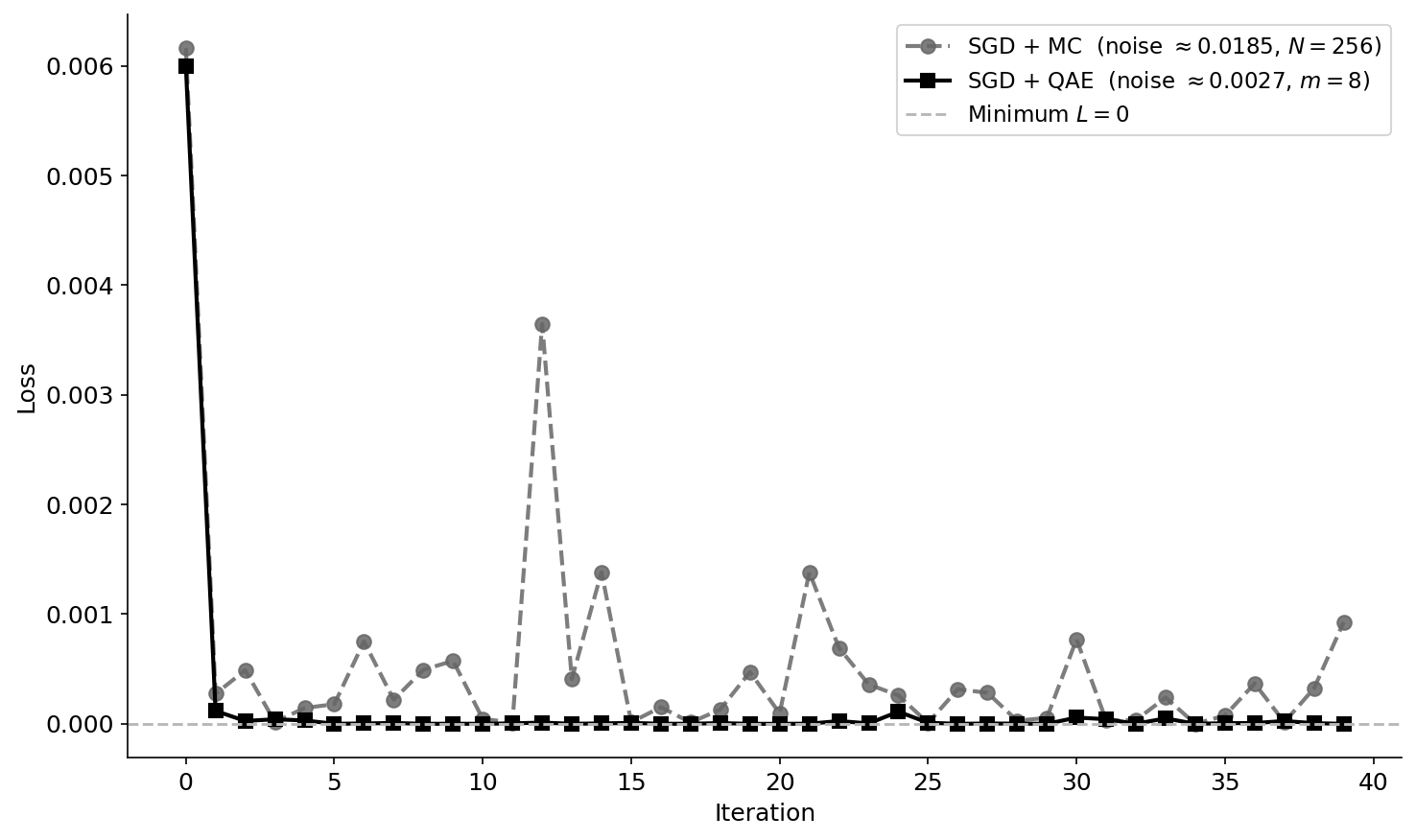}
\caption{SGD trajectories: MC vs.\ QAE.}\label{fig:sgd}
\end{figure}

Figure~\ref{fig:sgd} compares the two gradient descent trajectories over 40
iterations, obtained with the MC and QAE estimators. The MC path reaches the
minimum but continues to oscillate and never settles, since its
$\mathcal{O}(N^{-1})$ variance perturbs every step. The QAE path, instead,
converges within a few iterations and stays flat. The quadratic variance
reduction thus translates into a real qualitative difference at a fixed
computational budget.

\section{Conclusion}
\label{sec:conclusion}

We have demonstrated that the classical limitations associated with
gradient-based stochastic optimization can be overcome through the use of QAE.
Theorem~\ref{thm:main} and Corollary~\ref{cor:speedup} establish the quadratic
decrease of the gradient variance, with an advantage that grows linearly in
the computational budget, and the final simulation confirms it with
statistical significance.

\paragraph{Assumptions and limitations.}
Two assumptions restrict the scope of these results: the single-qubit model
and the hardware. The first one assumes that the state preparation operator
$\mathcal{A}$ is a trivial rotation; however, a real case, specifically a
multi-dimensional gradient, would require $\mathcal{A}$ to encode the full
data distribution in amplitude. This step relies on a quantum random access
memory (QRAM), which has yet to be practically realized~\cite{Aaronson2015}.
If the state preparation is costly, it may absorb the QAE speed-up and cancel
its advantage entirely. Second, QAE assumes a device capable of preserving
coherence over $\mathcal{O}(M)$ gates, but on NISQ hardware~\cite{Preskill2018}
gate noise degrades the advantage, though noise-resilient
variants~\cite{Suzuki2020} mitigate it. A natural extension is to place the
estimator in a variational quantum circuit~\cite{Cerezo2021}, where
$\mathcal{A}$ is the parametric circuit itself, for which the quantity to
estimate is already quantum and the QRAM problem disappears entirely.



\end{document}